\documentclass[aps,showpacs, showkeys,superscriptaddress,numerical,longbibliography,preprint]{revtex4-1}
\usepackage{amsfonts}
\usepackage{amsmath}
\usepackage{amssymb}
\usepackage{amsthm}
\usepackage[mathscr]{eucal}
\usepackage{url}
\usepackage[pdftex, colorlinks=true]{hyperref}
\hypersetup{pdfauthor={Jason Dominy and Herschel Rabitz}, pdftitle={Dynamic Homotopy and Landscape Dynamical Set Topology in Quantum Control}}

\DeclareMathOperator{\Tr}{Tr}

\DeclareMathOperator{\Span}{span}

\DeclareMathOperator{\ad}{ad}
\DeclareMathOperator{\sgn}{sgn}

\newtheorem{theorem}{Theorem}
\newtheorem{lemma}{Lemma}

\newcommand{\rmd}{\mathrm{d}}
\newcommand{\rmT}{\mathrm{T}}
\newcommand{\SUN}{\mathrm{SU}(N)}
\newcommand{\PUN}{\mathrm{PU}(N)}
\newcommand{\SO}{\mathrm{SO}}
\newcommand{\SU}{\mathrm{SU}}
\newcommand{\PU}{\mathrm{PU}}
\newcommand{\UN}{\mathrm{U}(N)}
\newcommand{\Un}{\mathrm{U}(\mathbf{n})}
\newcommand{\Um}{\mathrm{U}(\mathbf{m})}
\newcommand{\Unm}{\Un\oplus\Um}
\newcommand{\suN}{\mathrm{su}(N)}
\newcommand{\su}{\mathrm{su}}

\newcommand{\CPN}{\mathbb{C}\mathrm{P}^{N-1}}
\newcommand{\FM}[1][n_{1},\dots,n_{\kappa}]{\mathrm{Fl}({#1})}
\newcommand{\CTL}{\mathcal{E}}
\newcommand{\CTLB}{\mathcal{F}}
\newcommand{\CtlSp}{\mathbb{K}}
\newcommand{\EP}{\mathfrak{e}}
\newcommand{\MFD}{\mathcal{M}}
\newcommand{\VF}{f}
\newcommand{\Gr}[2]{\mathrm{Gr}_{#1}\big(\mathbb{C}^{#2}\big)}
\newcommand{\KSS}{\mathscr{A}} 

\begin{document}

\author{Jason Dominy}
\email{jdominy@usc.edu}
\altaffiliation[Current Address: ]{Center for Quantum Information Science \& Technology, University of Southern California, Los Angeles CA 90089}
\affiliation{Program in Applied and Computational Mathematics, Princeton University, Princeton NJ 08544}
\author{Herschel Rabitz}
\email{hrabitz@princeton.edu}
\affiliation{Department of Chemistry, Princeton University, Princeton NJ 08544}
\affiliation{Program in Applied and Computational Mathematics, Princeton University, Princeton NJ 08544}

\title{Dynamic Homotopy and Landscape Dynamical Set Topology in Quantum Control}

\date{ \today}

\begin{abstract}
We examine the topology of the subset of controls taking a given initial state to a given final state in quantum control, where ``state'' may mean a pure state $|\psi\rangle$, an ensemble density matrix $\rho$, or a unitary propagator $U(0,T)$.  The analysis consists in showing that the endpoint map acting on control space is a Hurewicz fibration for a large class of affine control systems with vector controls.  Exploiting the resulting fibration sequence and the long exact sequence of basepoint-preserving homotopy classes of maps, we show that the indicated subset of controls is homotopy equivalent to the loopspace of the state manifold.  This not only allows us to understand the connectedness of ``dynamical sets'' realized as preimages of subsets of the state space through this endpoint map, but also provides a wealth of additional topological information about such subsets of control space.
\end{abstract}



\maketitle

\section{Introduction}
The last few years have seen the concept of control landscapes used in connection with a variety of problems in quantum control theory \cite{Chakrabarti2007}.  Generally, a control landscape is a map, defined implicitly through the control dynamical system, that assigns a real observable value to each admissible control field.  In the case of quantum systems, the landscapes may be conveniently thought of as the composition of two maps: a control$\to$state map and a state$\to$observable map.  The ``state'' in this case is frequently the unitary time-evolution operator $U(T,0)$ which is the general solution to the Schr\"odinger equation at some final time $T$; so the control$\to$state map takes in a control field and a final time $T$ and produces $U(T,0)$.  The state$\to$observable map (often called the ``kinematic landscape'') is a real-valued function on the special unitary group $\SUN$, i.e. it takes in a unitary operator $U$ and produces the value of the final observable.  Under the composition control$\to$state$\to$observable (the so-called ``dynamical landscape''), the goal of optimal control is generally to maximize the observable with respect to the control.

Such landscape formulations have been used to gain some understanding of the critical point structure -- and thus the nature of the gradient flow -- of these optimal control problems for state-to-state transitions \cite{Rabitz2004, Rabitz2006, Hsieh2008}, general quantum mechanical observables on an ensemble \cite{Rabitz2006a, Ho2006, Wu2008}, and unitary transformation (quantum gate) preparation \cite{Rabitz2005, Hsieh2008a, Ho2009}.  In addition, these studies have considered the implications of the vast multiplicity of controls that are all capable of maximizing the landscape value.

In the present paper, we investigate one aspect of this multiplicity of solutions: the topologies of various important subsets of control space.  The main result presented here is a homotopy equivalence, under certain assumptions, between the set of controls carrying one point of a Riemannian manifold to another and the loop space of the reachable set.  A simple corollary of this result in the case of quantum control is that any connected subset of $\SUN$ (such as the maximal submanifold of any of the above kinematic landscapes) has a connected preimage in control space.  The full implications of the wealth of topological information offered by the homotopy equivalence are yet to be understood, but this corollary on connectedness is important in understanding the gross structure of control landscapes as well as the behavior of many optimization and exploration algorithms, especially continuous methods such as D-MORPH \cite{Rothman2005, Rothman2005a, Rothman2006, Beltrani2007, Dominy2008}.
  
This work has a relation to that of Colonius et. al. \cite{Colonius2005, Kizil2008,Kizil2009}, however it is much more closely aligned with the work of Sarychev \cite{Sarychev1991, Sarychev1991a}.  Indeed, the arguments presented in this paper follow very closely those of Sarychev.  Sarychev considered drift-free control problems on general Riemannian manifolds and proved a homotopy equivalence between the space of trajectories with the usual compact-open (i.e., uniform) topology and the loop space of the underlying Riemannian manifold.  We extend his results by incorporating a drift term and pulling the problem back to the control space with a stronger topology.  While we have presented this work in terms of quantum control, the principal theorem is proved for affine control problems with drift on general Riemannian manifolds.

To help place the present work in its proper context, it should be noted that the statement of the main theorem in Sarychev's second paper \cite{Sarychev1991a} on the subject is incorrect.  As stated, the theorem declares that the homotopy equivalence indicated above holds for any ``conic polysystem'' satisfying the Lie algebra rank condition.  If true, that theorem would largely contain the results of Theorem \ref{thm:homotopyEquiv} of the present work.  However, Sarychev's proof in \cite{Sarychev1991a} only covers \emph{symmetric} conic polysystems, which are systems without drift, suggesting that the wording of the theorem was a simple mis-statement.  Moreover, Montgomery \cite{Montgomery1992, Montgomery1996} has constructed a counterexample on $\mathrm{SO}(3)$ based on Little's \cite{Little1970} work on nondegenerate curves on $S^{2}$  which demonstrates that the conclusions of Sarychev's theorem do not hold for all conic polysystems.  However, Montgomery's example restricts the controls to be strictly positive functions.  The importance of this difference in the space of controls is made clear (as in Montgomery's paper \cite{Montgomery1996}) by comparing Little's theorem to Smale's earlier results on regular curves \cite{Smale1958}.  Translated into a scalar affine control problem with drift on $\mathrm{SO}(3)$, Smale's theorem applies to the case in which the control is an arbitrary continuous function of time and yields different topology to that of Little's theorem, which applies to the case where the control function is constrained to be positive valued.  In extending Sarychev's work, we consider a class of control problems with drift (thus generally outside the domain of the proofs offered by Sarychev) in which the space of controls is taken to be large enough to include all piecewise constant functions (thus excluding Montgomery's counterexample).

The paper is organized as follows.  Section \ref{sec:dynamicHomotopy} states and proves the main theorem of the paper: the homotopy equivalence between the fiber of the endpoint map and the loop space of the state manifold.  In section \ref{sec:quantumFiberTop}, this theorem is applied to various formulations of quantum control to determine the topologies of the corresponding fibers.  Methods of algebraic topology are employed in Section \ref{sec:dynamicalSetTopology} to bootstrap from the topologies of the fibers to the topologies of more general preimages of subsets of the state space.  The results are summarized in Section \ref{sec:conclusion}.  Five appendices are included that prove the continuity of various important maps, and describe some similar results that may be obtained for scalar control problems on 2-level systems (i.e. qubits).

\bigskip

\section{Dynamic Homotopy}
\label{sec:dynamicHomotopy}
Consider an affine control problem on an $n$ dimensional, second countable, smooth, connected Riemannian manifold $\MFD$, where the control system is of the form \begin{equation}\frac{\rmd z}{\rmd t} = \VF_{0}(z) + \sum_{i=1}^{m}\CTL_{i}(t)\VF_{i}(z) \qquad\qquad z(0) = z_{0}\in\MFD\label{eqn:affineControlSystem}\end{equation} and we assume the $\VF_{i}$'s are smooth, complete vector fields on $\MFD$ and that the Lie algebra generated by the control vector fields $\{\VF_{1},\dots,\VF_{m}\}$ (excluding the drift $\VF_{0}$) spans the tangent space of $\MFD$ at every point.  This latter assumption necessarily requires that $m$ be at least 2.  These assumptions imply that the system is strongly controllable \cite[Ch. 4, Thm. 2]{Jurdjevic1997}, i.e. that for any pair of points $z_{0},z_{f}\in\MFD$ and any $\hat{T}>0$, there exists a control $\vec{\CTL}$ that will steer the system from $z_{0}$ to $z_{f}$ in some time $0<T<\hat{T}$.  In laboratory implementations of quantum control, the $m=1$ case of scalar control is more common at present, but experiments with $m\geq 2$ can be implemented (see also Appendix \ref{sec:dynHomQubits}).

Now, let $\CtlSp_{0}$ denote the space of $(m+1)$-tuples $(T,\CTL_{1},\dots,\CTL_{m})$ consisting of a final time $T\geq 0$ and $m$ functions in $L^{1}[0,\infty)$ such that $\CTL_{1}(t) = \cdots = \CTL_{m}(t) = 0$ for all $t>T$.  This space will be topologized by the metric 
\begin{widetext}
	\begin{equation}d\big((T_{1},\CTL_{1},\dots, \CTL_{m}), (T_{2}, \CTLB_{1},\dots,\CTLB_{m})\big) = |T_{1}-T_{2}| + \sum_{j=1}^{m}\int_{0}^{\infty}|\CTL_{j}(t)-\CTLB_{j}(t)|\,\rmd t.\end{equation}
\end{widetext}
This $L^{1}$ topology is somewhat weaker than  the corresponding $L^{2}$ topology, but will be necessary for continuity of a certain map later on.  We can also define a concatenation operation on $\CtlSp_{0}$, denoted by $\star$, so $(T_{1}, \CTL_{1},\dots,\CTL_{M})\star(T_{2},\CTLB_{1},\dots,\CTLB_{M}) = (T_{1}+T_{2},\tilde{\CTL}_{1},\dots,\tilde{\CTL}_{2})$, where \begin{equation}\tilde{\CTL}_{k}(t) = \begin{cases}\CTL_{k}(t) & 0\leq t < T_{1}\\\CTLB_{k}(t - T_{1}) & T_{1}\leq t\leq T_{1}+T_{2}\\0 & \text{else}.\end{cases}\label{eqn:controlConcat}\end{equation}

Let $\CtlSp$ denote a subset of $\CtlSp_{0}$ that contains the piecewise constant controls, is closed under concatenation, and is such that for every element of $\CtlSp$, the system \eqref{eqn:affineControlSystem} admits a unique trajectory over the corresponding interval $[0,T]$ for any initial point $z_{0}\in\MFD$ (any such subset will suffice).  We will assume also that if $C = (T,\CTL_{1},\dots,\CTL_{m})\in\CtlSp$, then for any $0\leq T'\leq T$, the truncation of $C$ to the interval $[0,T']$ is also an element of $\CtlSp$.  Now, for any $x\in\MFD$, define a map $\EP_{x}:\CtlSp\to\MFD$ in the following way.  For a control $C = (T,\mathcal{E}_{1},\dots,\mathcal{E}_{m})\in\CtlSp$, integrate the control system \eqref{eqn:affineControlSystem} out to time $T$ with initial condition $z_{0} = x$ and define $\EP_{x}(C) := z(T)$, the solution at the final time $T$ for control $C$.  This $\EP_{x}$ will be called the \emph{endpoint map} on the control space.  It is a continuous map, as shown in Appendix \ref{sec:contEndpoint}.  In addition $\EP_{x}(C_{1}\star C_{2}) = \EP_{\EP_{x}(C_{1})}(C_{2})$.

As indicated in the introduction, the principal result in this paper is the following theorem in which $P\MFD$ is the path space on the state manifold $\MFD$:
\begin{theorem}
For each $x,y\in\MFD$, the subset $\EP_{x}^{-1}(y)\subset \CtlSp$ of controls carrying $x$ to $y$ is homotopy equivalent to $\Omega\MFD$, the loop space of the manifold $\MFD$.  This homotopy equivalence is carried by the restriction of the trajectory map $\tau:\CtlSp\to PB$, defined implicitly through the control system \eqref{eqn:affineControlSystem} with initial condition $x$.\label{thm:homotopyEquiv}
\end{theorem}
For the sake of clarity, fix now some basepoint $b\in\MFD$. The proof of this theorem will involve first showing that the endpoint map $\EP_{b}$ is a (Hurewicz) fibration and then exploiting the resulting fibration sequence and corresponding long exact sequence of base-point preserving homotopy spaces from some arbitrary pointed topological space.

\subsection{The Endpoint Map is a Fibration}
We begin with the proof that $\EP_{b}$ is a fibration \cite{Hurewicz1955,Hatcher2002}.  For a continuous map $p:E\to B$ between topological spaces, let $X_{p}\subset E\times PB$ consist of all pairs $(e, \gamma)$ such that $p(e) = \gamma(0)$, where $PB$ is the path space on B with compact-open topology, and let $\tilde{p}:PE\to X_{p}$ be the map $\tilde{p}(\eta) = (\eta(0), p\eta)$.  This map $p$ is called a (Hurewicz) fibration if there exists a continuous \emph{lifting function} $\lambda:X_{p}\to PE$ such that $\tilde{p}\circ\lambda$ is the identity on $X_{p}$.  A continuous map $p:E\to B$ is a \emph{local} fibration if for each $b\in B$ there exists an open neighborhood $G$ such that $p$ restricted to $p^{-1}(G)$ is a fibration.  The Hurewicz uniformization theorem \cite{Hurewicz1955} states that if $p$ is a local fibration and $B$ is paracompact, then $p$ is a (global) fibration.  Since the endpoint map $\EP_{b}:\CtlSp\to\MFD$ maps to a manifold, the base space is paracompact, so we need only prove that $\EP_{b}$ is a local fibration to have the desired conclusion.  To that end, we first prove the existence of a  particular cross-section map on a neighborhood of $x_{0}\in\MFD$.

\begin{lemma}[The Cross-Section Map]For any $x_{0}\in\MFD$, there exist an open neighborhood $\mathcal{W}$ of $x_{0}$ in $\MFD$ and a continuous function $\sigma:\mathcal{W}\times\mathcal{W}\to\CtlSp$ such that (1) $\EP_{x}(\sigma(x,y)) = y$ for all $x,y\in\mathcal{W}$ (i.e. $\sigma(x,y)$ is a control taking $x$ to $y$) and (2) for all $x\in\mathcal{W}$, $\sigma(x,x)$ is the zero time control with $T=0$ and $\CTL_{i} \equiv 0$ for all $i$.\label{lem:CrossSectionMap}
\end{lemma}
\begin{proof}
Given a sequence of generators $Y_{1},\dots, Y_{k}$, define the a sequence of compositions of formal exponentials by
\begin{widetext}
\begin{subequations}
\begin{align}
Q^{1}(Y_{1}) & := e^{Y_{1}}\\
Q^{2}(Y_{1},Y_{2}) & := e^{Y_{2}}\circ e^{Y_{1}}\circ e^{-Y_{2}}\circ e^{-Y_{1}}\\
\vdots & \nonumber\\
Q^{\nu}(Y_{1},\dots,Y_{\nu}) & := e^{Y_{\nu}}\circ Q^{\nu-1}(Y_{1},\dots,Y_{\nu-1})\circ e^{-Y_{\nu}}\circ \big(Q^{\nu-1}(Y_{1},\dots,Y_{\nu-1})\big)^{-1}.
\end{align}
\end{subequations}
\end{widetext}
If $\{Y_{1},\dots,Y_{\nu}\}$ are sufficiently close to zero, then using the Baker-Campbell-Hausdorff (BCH) formula \cite{Hall2003}, it is easy to see that $Q^{2}(Y_{1},Y_{2}) = e^{[Y_{2},Y_{1}] + \{\text{terms of order} > 2\}}$.  And if $Q^{\nu-1}(Y_{1},\dots, Y_{\nu-1}) = e^{\ad Y_{\nu-1}\dots\ad Y_{2}Y_{1} + \{\text{terms of order}> \nu-1\}}$, then the BCH formula implies that $Q^{\nu}(Y_{1},\dots, Y_{\nu}) = e^{\ad Y_{\nu}\dots \ad Y_{2}Y_{1} + \{\text{terms of order} > \nu\}}$.  By induction, this is true for all $\nu = 1,2,\dots$.

Now, define a corresponding sequence of maps $R^{\nu}(Y_{0},Y_{1},\dots,Y_{\nu},\xi_{1},\dots,\xi_{\nu})$, where $\{\xi_{i}\}\in\mathbb{R}$, by taking the sequence of exponentials in $Q^{\nu}(Y_{1},\dots,Y_{\nu})$ (\emph{after} having carried out the recursion in the above definition) and replacing $e^{\pm Y_{k}}$ by $e^{\xi_{k}^{2\alpha_{\nu}}Y_{0} \pm \xi_{k}Y_{k}}$ for some positive integer $\alpha_{\nu}$.  Then $R^{\nu}$ can be written, using BCH, as the exponential of a series of terms drawn from $\{\xi_{1}^{2\alpha_{\nu}}Y_{0}, \dots, \xi_{\nu}^{2\alpha_{\nu}}Y_{0}, \xi_{1}Y_{1}, \dots, \xi_{\nu}Y_{\nu}\}$ and the Lie brackets thereof.  Among the terms in this sum that do not include $Y_{0}$, the lowest degree in the $\xi_{k}$'s is $\xi_{1}\cdots \xi_{\nu}\ad Y_{\nu}\dots\ad Y_{2}Y_{1}$, while the lowest degree term that includes $Y_{0}$ is degree $2\alpha_{\nu}$ in the $\xi_{k}$'s.  So if we choose $\alpha_{\nu} = \lfloor \nu/2 +1 \rfloor > \nu/2$, then 
\begin{widetext}
\begin{equation}R^{\nu}(Y_{0},Y_{1},\dots,Y_{\nu},\xi_{1},\dots,\xi_{\nu}) = e^{\xi_{1}\cdots\xi_{\nu}\ad Y_{\nu}\dots\ad Y_{2}Y_{1} + \{\text{terms of degree } > \nu \text{ in } \{\xi_{i}\}\}},\end{equation}
\end{widetext}
where \emph{degree} refers to the polynomial degree in the variables $\{\xi_{1},\dots,\xi_{\nu}\}$, rather than the number of Lie brackets or vector fields involved.

Using the fact that the Lie algebra generated by the control vector fields $\{\VF_{j}\}_{j=1}^{m}$ from \eqref{eqn:affineControlSystem} spans the entire tangent space $\rmT_{x}\MFD$ for each $x$, fix $\{X_{kj}\}\subset \Span\{\VF_{1},\dots, \VF_{m}\}$ such that $\{\Xi_{k}:=\ad X_{k\nu_{k}}\dots \ad X_{k2}X_{k1}\}$ for $k=1,\dots, n$ forms a basis for $\rmT_{x_{0}}\MFD$ and define $F_{1}:\mathcal{W}\times \mathbb{R}^{n}\to\MFD$ by 
\begin{widetext}
\begin{subequations}
\begin{align}F_{1}(x,r) & := \left(\prod_{k=1}^{n}R^{\nu_{k}}\left(\VF_{0}, X_{k1},X_{k2},\dots, X_{k\nu_{k}}, \sgn r_{k}|r_{k}|^{1/\nu_{k}}, |r_{k}|^{1/\nu_{k}}, \dots, |r_{k}|^{1/\nu_{k}}\right)\right)(x)\\
& = e^{r_{n}\Xi_{n} + \{\text{higher order terms in } r_{n}\}}\cdots e^{r_{1}\Xi_{1} + \{\text{higher order terms in } r_{1}\}}(x)
\end{align}
\end{subequations}
\end{widetext}
where the ``higher order terms in $r_{k}$'' each have $r_{k}$ dependence of either $|r_{k}|^{1+j/\nu_{k}}$ or $r_{k}|r_{k}|^{j/\nu_{k}}$ for some positive integer $j$, so the higher order terms are $C^{1}$ functions of $r_{k}$ with derivative zero at $r_{k}=0$.
Then $F_{1}$ is $C^{1}$, $F_{1}(x, 0) = x$ for all $x\in\MFD$, and $\frac{\partial F_{1}}{\partial r_{k}}(x_{0},0) = \Xi_{k}(x_{0}) = \ad X_{k\nu_{k}}\dots \ad X_{k2}X_{k1}(x_{0})$.  Because these nested Lie brackets were chosen to be linearly independent, the Jacobian operator $\frac{\partial F_{1}}{\partial r}(x_{0},0):\mathbb{R}^{n}\to\rmT_{x_{0}}\MFD$ is full rank, hence an isomorphism.  So by the implicit function theorem, there exists a neighborhood $\mathcal{W}\subset\MFD$ of $x_{0}$ and a unique differentiable (i.e., $C^{1}$) map $\varphi:\mathcal{W}\times\mathcal{W}\to\mathbb{R}^{n}$ such that $\varphi(x,x) = 0$ and $F_{1}(x,\varphi(x,y)) = y$ for all $x,y\in\mathcal{W}$.

Now, observe that the sequence of exponentials in $R^{\nu}(Y_{0}, Y_{1},\dots,Y_{\nu}, \xi_{1},\dots,\xi_{\nu})$ can be interpreted as the solution of the control system \eqref{eqn:affineControlSystem} with piecewise constant controls.  For example, in the case of $R^{2}(\VF_{0}, Y_{1},Y_{2},\xi_{1},\xi_{2})$, where $Y_{j} = \sum_{k=1}^{m}\kappa_{jk}\VF_{k}$, the control $C\in\CtlSp$ has final time $T =  2\xi_{1}^{4} + 2\xi_{2}^{4}$ and control functions
\begin{equation}
\CTL_{k}(t) = \begin{cases}-\kappa_{1k}\xi_{1}^{-3} & 0<t<\xi_{1}^{4}\\ 
	-\kappa_{2k}\xi_{2}^{-3} & \xi_{1}^{4}<t<\xi_{1}^{4} + \xi_{2}^{4}\\
	+\kappa_{1k}\xi_{1}^{-3} & \xi_{1}^{4} + \xi_{2}^{4}<t<2\xi_{1}^{4} + \xi_{2}^{4}\\
	+\kappa_{2k}\xi_{2}^{-3} & 2\xi_{1}^{4} + \xi_{2}^{4}<t<2\xi_{1}^{4} + 2\xi_{2}^{4}.
\end{cases}
\end{equation} 
These controls can be shown to vary continuously with $\{\xi_{k}\}$ in the metric established for the control space $\CtlSp$ (see Appendix \ref{sec:contCrossSec}).  Now define $F:\mathbb{R}^{n}\to \CtlSp$ to be the map $\EP_{x}\big(F(r)\big) = F_{1}(x,r)$ that maps $r\in \mathbb{R}^{n}$ to the piecewise constant control in $\CtlSp$ whose integral curve corresponds exactly to the sequence of exponentials in the definition of $F_{1}(x,r)$.  Then $F$ is a continuous map, and $\sigma:\mathcal{W}\times\mathcal{W}\to\CtlSp$ defined by $\sigma(x,y) = F(\varphi(x,y))$ is also a continuous map.  In addition, since $\varphi(x,x) = 0$, $\sigma(x,x)$ is the zero time control, because all of the $\xi_{k}$'s are zero, meaning that the final time $T$ is zero.  The map $\sigma$ is therefore exactly the desired cross-section map.
\end{proof}

We now return to the definition of a local fibration and define a function $\lambda_{\mathcal{W}}:X_{\mathcal{W}}\to P\CtlSp$, where $X_{\mathcal{W}}\subset \EP_{b}^{-1}(\mathcal{W})\times P\mathcal{W}$ is the set of all $(C,\gamma)$ such that $\EP_{b}(C) = \gamma(0)$, by letting $\lambda_{\mathcal{W}}(C,\gamma)(s) := C\star\sigma(\EP_{b}(C),\gamma(s))$.  Since $\lambda_{\mathcal{W}}$ is continuous (see Appendix \ref{sec:contLifting}) and $\big(\EP_{b}|_{\EP^{-1}(X_{\mathcal{W}})}\big)\circ\lambda_{\mathcal{W}}$ is the identity on $X_{\mathcal{W}}$, $\lambda_{\mathcal{W}}$ is a lifting function for $\EP_{b}\big|_{\EP_{b}^{-1}(\mathcal{W})}$.  Therefore $\EP_{b}$ is a local Hurewicz fibration and, by the Hurewicz uniformization theorem, also a global fibration.  Note that $\lambda_{\mathcal{W}}$ has the additional property that for a concatenation of unitary paths $\gamma_{1}\star\gamma_{2}$, $\lambda_{\mathcal{W}}(C,\gamma_{1}\star\gamma_{2}) = \lambda_{\mathcal{W}}(\lambda_{\mathcal{W}}(C,\gamma_{1}),\gamma_{2})$.

\subsection{Topology of the Fiber}
In order to establish the homotopy equivalence promised by Theorem \ref{thm:homotopyEquiv}, the following lemma is needed.
\begin{lemma}
\label{lemma:fiberHomotopyEquiv}
For any fibration $p:E\to B$ with contractible total space $E$ and path connected base $B$, the map $\eta:\Omega B\to F = p^{-1}(b_{0})$ given by $\eta(\gamma) = \lambda(e_{0},\gamma)(1)$ is a homotopy equivalence.  Moreover, \emph{any} map $F\to\Omega B$ that extends to a map $g:E\to PB$ such that $g(e)(1) = p(e)$ is a homotopy inverse of $\eta$, hence also a homotopy equivalence.
\end{lemma}
\begin{proof}
Consider the homotopy $H:PB\times [0,1]\to E$ given by $H_{s}(\gamma) = \lambda(e_{0},\gamma)(s)$, where $e_{0}$ is the basepoint of $E$ and $\lambda$ is the lifting function of $p$.  Let $\rho:PB\to B$ be the endpoint map on the path space $PB$, so that $\rho$ is also a fibration.  Then $H_{1}$ satisfies $p\circ H_{1} = \rho$, so that $H_{1}$ is a fiber-preserving map, and $H_{0}$ is the constant map $H_{0}(\gamma)= b_{0}$ for all $\gamma\in PB$.  Let $g:E\to PB$ be any fiber-preserving map.  Then $g\circ H_{1}\simeq g\circ H_{0}$ which is a constant map on $PB$, and $H_{1}\circ g\simeq H_{0}\circ g$ which is a constant map on $E$.  Since $PB$ and $E$ are both contractible by assumption, $g\circ H_{1}$ and $H_{1}\circ g$ are homotopic to the identity maps on $PB$ and $E$, respectively, so that $H_{1}$ and $g$ are homotopy equivalences.  Moreover, since they are both fiber-preserving, they are fiber homotopy equivalences \cite[\S 4H, Ex. 3]{Hatcher2002}\cite[\S 7.5]{May1999}, and fiber homotopy inverses of one another.  Letting $F$ denote the canonical fiber $F:=p^{-1}(b_{0})$ and noting that the corresponding fiber of $\rho$ is $\Omega B = \rho^{-1}(b_{0})$, this implies that the restrictions of $H_{1}$ to $\Omega B$ and $g$ to $F$ are homotopy equivalences between these two spaces, and homotopy inverses of one another.
\end{proof}

In the case of the endpoint map of the control system $\EP_{b}:\CtlSp\to\MFD$, for each $s\in[0,1]$, let $\rho_{s}:\CtlSp\to\CtlSp$ be the map $\rho_{s}(T,\CTL_{1},\dots,\CTL_{M}) = ((1-s)T, \tilde{\CTL}_{1},\dots, \tilde{\CTL}_{M})$, where $\tilde{\CTL}_{j}(t) = \CTL_{j}(t)$ on $[0,(1-s)T]$ and is zero after $(1-s)T$.  Then $\rho_{s}(C)$ is continuous in $C$ and $s$, $\rho_{0}$ acts as the identity map on $\CtlSp$, $\rho_{s}(0,0,\dots,0) = (0,0,\dots,0)$, and  $\rho_{1}$ is the trivial map sending every control in $\CtlSp$ to the zero time control $(0,0,\dots,0)$.  So $\rho_{s}$ is a deformation retract of $\CtlSp$ to the point $(0,0,\dots,0)$.  Therefore $\CtlSp$ is contractible.  Since $\MFD$ is assumed to be connected, Lemma \ref{lemma:fiberHomotopyEquiv} implies that $\eta:\Omega\MFD\to F:=\EP_{b_{0}}^{-1}(b_{0})$ is a homotopy equivalence, where $\eta(\gamma) = \lambda(C_{0},\gamma)(1)$, $C_{0}$ is the zero time control, and $F\subset\CtlSp$ is the set of controls taking the basepoint $b_{0}\in\MFD$ back to $b_{0}$.  Let $\tau:\CtlSp\to PB$ be the trajectory map, defined implicitly through the control system \eqref{eqn:affineControlSystem} with initial condition $b_{0}$, that integrates the system to obtain the trajectory and then linearly reparametrizes the trajectory to obtain a curve in $PB$.  Then $p\circ\tau = \rho$, so that $\tau$ is fiber-preserving, so that, again by Lemma \ref{lemma:fiberHomotopyEquiv}, $\tau$ restricted to $F$ is a homotopy equivalence between $F$ and $\Omega\MFD$.  In addition, all fibers of $\EP_{b_{0}}$ are homotopy equivalent \cite{Hatcher2002} and the base point $b_{0}$ is arbitrary, so Theorem \ref{thm:homotopyEquiv} is proved.  Furthermore, since $\MFD$ is a second countable, hence separable, manifold, both $\MFD$ and $\Omega\MFD$ have the homotopy type of countable CW-complexes \cite{Milnor1959} and therefore the fibers $\EP_{x}^{-1}(y)$ all have the homotopy type of countable CW-complexes.

Since $\EP_{b}:\CtlSp\to\MFD$ is a Hurewicz fibration with canonical fiber $F:=\EP_{b}^{-1}(b)$, it admits a fibration sequence 
\begin{widetext}
\begin{equation}\cdots \to \Omega^{2}\CtlSp \to\Omega^{2}\MFD\to\Omega F\to\Omega\CtlSp\to\Omega\MFD\stackrel{\eta}{\to} F\stackrel{i}{\to}\CtlSp\stackrel{\EP_{b}}{\to}\MFD\label{eqn:fibrationSequence}\end{equation}
\end{widetext}
where each triple $X\to Y\to Z$ is a fibration $Y\to Z$ with fiber $X$ up to homotopy equivalence \cite{Hatcher2002}.  The map $i:F\to\CtlSp$ above is the inclusion map, and the map $\eta:\Omega\MFD\to F$ is defined as in Lemma \ref{lemma:fiberHomotopyEquiv} \cite{Hatcher2002}.  For any pointed topological space $(Z,z_{0})$, this fibration sequence gives rise to a long exact sequence of spaces of basepoint-preserving homotopy classes of basepoint-preserving maps \cite{Hatcher2002}
\begin{widetext}
\begin{equation}\to\!\langle Z,\Omega^{2}\CtlSp\rangle \!\to\!\langle Z, \Omega^{2}\MFD\rangle\!\to\!\langle Z, \Omega F\rangle\!\to\!\langle Z, \Omega\CtlSp\rangle\!\to\!\langle Z, \Omega\MFD\rangle\!\stackrel{\eta_{*}}{\to}\! \langle Z, F\rangle\!\stackrel{i_{*}}{\to}\!\langle Z, \CtlSp\rangle\!\stackrel{{\EP_{b}}_{*}}{\to}\!\langle Z, \MFD\rangle.\label{eqn:longExactFibrationSequence}\end{equation}
\end{widetext}
The contractibility of $\CtlSp$ implies that $\Omega\CtlSp$ is also contractible, so that $\langle Z,\Omega\CtlSp\rangle = \langle Z, \CtlSp\rangle = 0$ for all $(Z,z_{0})$.  Then the long exact sequence \eqref{eqn:longExactFibrationSequence} contains the short exact sequence \begin{equation}0 \to \langle Z, \Omega\MFD\rangle \stackrel{\eta_{*}}{\longrightarrow} \langle Z, F\rangle \to 0.\label{eqn:shortExactSequence}\end{equation} 

Since control concatenation is continuous and the zero time control acts as the identity  with respect to concatenation (i.e. $0\star C = C\star 0 = C$), $F$ is an H-space \cite{Hatcher2002}.  This structure on $F$ induces a natural monoid structure on $\langle Z, F\rangle$.  In addition, the natural H-group structure on $\Omega\MFD$ induces a group structure on $\langle Z, \Omega\MFD \rangle$.  Since, as we noted above, the lifting function respects path concatenation, so too will $\eta$, i.e. $\eta(\omega_{1}\star\omega_{2}) = \eta(\omega_{1})\star\eta(\omega_{2})$ for two loops $\omega_{1},\omega_{2}\in\Omega\MFD$, and therefore $\eta_{*}$ will respect the natural products on $\langle Z, \Omega\MFD\rangle$ and $\langle Z, F\rangle$.  Also note that the trivial loop (constant loop at $\mathbb{I}$) in $\Omega\MFD$ maps to the zero time control in $F$.  Then the fact that $\eta_{*} $ is a bijection (due to the exactness of \eqref{eqn:shortExactSequence}) means that the induced monoid structure on $\langle Z, F\rangle$ is, in fact, a group structure and $\eta_{*}$ is a group isomorphism.  

An obvious consequence is that all of the standard homotopy groups are isomorphic for these spaces.  Indeed, $\pi_{i}(\EP_{x}^{-1}(y)) := \langle S^{i}, \EP_{x}^{-1}(y)\rangle \simeq \langle S^{i}, \Omega\MFD\rangle =: \pi_{i}(\Omega\MFD) \simeq \pi_{i+1}(\MFD)$ for each $i\geq 0$.  In addition to this homotopy structure, the homotopy equivalence between $\EP_{x}^{-1}(y)$ and $\Omega\MFD$ implies that these two spaces share isomorphic homology groups [$H_{n}(\EP_{x}^{-1}(y); G)\simeq H_{n}(\Omega\MFD; G)$] and cohomology groups [$H^{n}(\EP_{x}^{-1}(y); G)\simeq H^{n}(\Omega\MFD; G)$] for all coefficient groups $G$ \cite[Prop. 4.21]{Hatcher2002}.  We will take a closer look at these topological indices for the common cases of quantum control in the next section.

\bigskip

\section{Fiber Topologies in Quantum Control}
\label{sec:quantumFiberTop}
We now turn our attention back to quantum control and consider the detailed topologies of the fibers in the most common formulations, where either the manifold $\MFD$ is the special unitary group $\SUN$ representing the space of unitary time propagation operators that are the general solutions of the Schr\"odinger equation, or the manifold $\MFD$ is the complex projective space $\CPN$ of states (modulo phase) of an $N$-level quantum system, or the manifold $\MFD$ is a generalized complex flag manifold $\FM$ of unitarily accessible density matrices.  As shown in Theorem \ref{thm:homotopyEquiv}, this reduces to understanding the topology of the loopspaces $\Omega\SUN$, $\Omega\CPN$, and $\Omega \FM$.

\subsection{Special Unitary Group}
First consider the case where the control problem \eqref{eqn:affineControlSystem} is a right-invariant affine system on $\SUN$, i.e. the Schr\"odinger equation 
\begin{equation}
	i\hbar \frac{d}{dt}U(t) = H(t)U(t) \qquad U(0) = \mathbb{I} \qquad H(t)  = H_{0} + \sum_{i=1}^{m}\CTL_{i}(t)H_{i}\label{eqn:controlHamiltonian}
\end{equation}
where the $H_{i}$'s are Hermitian and traceless and the control terms $\{iH_{1},\dots, iH_{m}\}$ generate the Lie algebra $\suN$.  Theorem \ref{thm:homotopyEquiv} states that for any target unitary propagator $W\in\SUN$, the set of controls in $\CtlSp$ that take the identity $\mathbb{I}$ to $W$, $\EP_{\mathbb{I}}^{-1}(W)$, is homotopy equivalent to $\Omega\SUN$.  Now, the homotopy groups of $\SUN$ can be shown by Bott periodicity \cite{Bott1959,Milnor1973} to be $\pi_{i}(\MFD) = 0$ for $0\leq i\leq 2N-1$ and $i$ even, $\pi_{1}(\MFD) = 0$, and $\pi_{i}(\MFD) \simeq \mathbb{Z}$ for $2\leq i\leq 2N-1$ and $i$ odd.  Therefore, \begin{equation}\pi_{i}(\EP_{\mathbb{I}}^{-1}(W)) \simeq \begin{cases}0 & i = 0\\0 & 1\leq i\leq 2N-3 \text{ and } i \text{ odd}\\ \mathbb{Z} & 2\leq i\leq 2N-2 \text{ and } i \text{ even.}\end{cases}\end{equation}  As an immediate consequence, $\EP_{\mathbb{I}}^{-1}(W)$ is both connected and simply connected for all $N\geq 2$.

The integral cohomology of $\Omega\SUN$ was also investigated by Bott in the 1950's \cite{Bott1956, Bott1958a, Milnor1973} and shown to be $H^{*}(\Omega\SUN; \mathbb{Z}) \simeq \mathbb{Z}[x_{2},x_{4},\dots,x_{2N-2}]$, the polynomial ring over $\mathbb{Z}$, where the $x_{k}$ are generators of degree $k$.  Importantly, this implies that the cohomology of $\Omega\SUN$ is trivial at odd dimensions.

\bigskip

\subsection{Complex Projective Space}
The space of pure quantum states of an $N$-level system may be modeled as the complex projective space $\CPN$, described by the fiber bundle $S^{1}\to S^{2N-1}\to\CPN$, which gives rise to the long exact sequence
\begin{widetext}
\begin{equation}\cdots\to\pi_{2}(S^{1})\to\pi_{2}(S^{2N-1})\to\pi_{2}(\CPN)\to\pi_{1}(S^{1})\to\pi_{1}(S^{2N-1})\to\pi_{1}(\CPN)\to 0.\end{equation}
\end{widetext}
Since $\pi_{i}(S^{1}) = 0$ for all $i\neq 0$, this implies that $\pi_{i}(\CPN) \simeq \pi_{i}(S^{2N-1})$ for all $i>2$.  Also, since $\pi_{i}(S^{2N-1}) = 0$ for all $i<2N-1$, we have that $\pi_{1}(\CPN) = 0$ and $\pi_{2}(\CPN) \simeq \mathbb{Z}$.  Therefore we find that for any two states $\psi,\phi\in\CPN$, the first $2N-2$ homotopy groups of $\EP_{\psi}^{-1}(\phi)$ are \begin{equation}\pi_{i}(\EP_{\psi}^{-1}(\phi))\simeq\begin{cases}0 & i = 0\\ \mathbb{Z} & i = 1 \\ 0 & 2\leq i \leq 2N-3\\\mathbb{Z} & i = 2N-2.\end{cases}\label{eqn:CPNfiberHomotopy}\end{equation}

The integral cohomology of $\Omega\CPN$ was also investigated in the 1950's \cite{Bott1958, Halpern1958} and shown to be $H^{*}(\Omega\CPN; \mathbb{Z})\simeq \bigwedge[x_{1}]\otimes\mathbb{Z}[y_{1},y_{2},y_{3},\dots]/\{i!j!y_{i}y_{j} = (i+j)!y_{i+j}\}$ where $x_{1}$ is a generator of degree 1 and $y_{j}$ is a generator of degree $j*(2n-2)$ for each $j = 1, 2, \dots$.

\bigskip

\subsection{Generalized Complex Flag Manifolds}
The closed-system evolution of a density matrix $\rho$ (a positive, trace one, Hermitian matrix), can be described by the von Neumann equation \begin{equation}i\hbar\frac{d}{dt}\rho(t) = [H(t),\rho(t)] \qquad \qquad \rho(0) = \rho_{0} \qquad H(t)  = H_{0} + \sum_{i=1}^{m}\CTL_{i}(t)H_{i}\end{equation} or, equivalently, $\rho(t) = U(t)\rho_{0}U^{\dag}(t)$, where $U(t)\in\SUN$ is the unitary propagator described in \eqref{eqn:controlHamiltonian}.  Given some fixed initial condition $\rho_{0}$, the space of density matrices reachable by closed-system dynamics [assuming $\{iH_{j}\}$ generates $\suN$] is the orbit $\{U\rho_{0}U^{\dag}\;:\; U\in\SUN\}$, which can be regarded as the generalized complex flag manifold $\FM \simeq \SUN/\mathrm{S}\big(\mathrm{U}(n_{1})\oplus\dots\oplus \mathrm{U}(n_{\kappa})\big)$, where the $n_{j}$'s are the multiplicities of the eigenvalues of $\rho_{0}$.  In the case where $\rho_{0}$ is a pure state (i.e., $\rho_{0} = |\psi_{0}\rangle\langle\psi_{0}|$), this is just the complex projective space, i.e. $\FM[1,N-1]\simeq \CPN$.  At the other extreme, if $\rho_{0}$ is completely non-degenerate (representing a Boltzmann distribution over non-degenerate energy levels, for example), then the state space is the \emph{complete} flag manifold $\FM[1,1,\dots,1]$.

These flag manifolds are compact, algebraic, homogeneous K\"ahler manifolds \cite{Borel1954}, simultaneously possessing natural structures of complex analytic manifolds with Hermitian metrics, Riemannian manifolds, symplectic manifolds, and algebraic varieties.  Topologically, they share the properties of being connected and simply connected, i.e. $\pi_{0}\big(\FM\big) \simeq \pi_{1}\big(\FM\big) \simeq 0$.  Also, from the homotopy long exact sequence of the fibration \begin{equation}\mathrm{S}\big(\mathrm{U}(n_{1})\oplus\dots\oplus \mathrm{U}(n_{\kappa})\big) \to \SUN \to \FM,\end{equation} it may be seen that $\pi_{2}\big(\FM\big) \simeq \pi_{2}\big(\mathrm{S}\big(\mathrm{U}(n_{1})\oplus\dots\oplus \mathrm{U}(n_{\kappa})\big)\simeq \mathbb{Z}^{\kappa -1}.$  Therefore the loop spaces are all connected, $\pi_{0}\big(\Omega\FM\big) \simeq \pi_{1}\big(\FM\big) \simeq 0$, and have nontrivial fundamental group: $\pi_{1}\big(\Omega\FM\big) \simeq\pi_{2}\big(\FM\big) \simeq \mathbb{Z}^{\kappa-1}$.  Higher order homotopy information about these spaces will depend on the choices of the parameters $(n_{1},\dots,n_{\kappa})$, i.e. on the eigenstructure of $\rho_{0}$.

The integral homology ring of the loop space of the complete flag manifold was just recently shown to be 
\begin{widetext}
\begin{align}
	\lefteqn{H_{*}\big(\Omega\FM[1,1,\dots,1]; \mathbb{Z}\big) \simeq H_{*}\big(\Omega\big(\SUN/T^{N-1}\big) ; \mathbb{Z}\big)}&\nonumber\\
	& \quad \simeq \big(T(x_{1},\dots,x_{N-1})\otimes\mathbb{Z}[y_{1},\dots,y_{N-1}]\big)/\langle x_{k}^{2} = x_{p}x_{q} = 2y_{1} \text{ for } 1\leq k,p,q\leq N-1, \; p\neq q\rangle,
\end{align}
\end{widetext}
where the generators $x_{j}$ are of degree $1$, and the generators $y_{j}$ are of degree $2j$, and $T(x_{1},\dots,x_{N-1})$ is the tensor algebra generated by $x_{1},\dots,x_{N-1}$ \cite{Grbic2010}.  The \emph{integral} (co)homologies of the loop spaces of the other generalized complex flag manifolds appear not to have been published, but both the  rational and mod $p$ cohomologies have been investigated \cite{Smith1968}.

\bigskip

\section{Topology of Landscape Level Sets and Critical Sets in Quantum Control}
\label{sec:dynamicalSetTopology}
In the study of quantum control landscapes, we are often interested in ``dynamical'' subsets $\EP_{x}^{-1}(\KSS)\subset \CtlSp$ corresponding to ``kinematic'' subsets $\KSS\subset\MFD$.  Most typically, these kinematic subsets are level sets and critical sets of objective functions defined on $\MFD$.  In this section, we consider the tools that may be used to investigate the topology of these types of dynamical subsets and analyze a few examples from quantum control.  In particular, we will examine the topologies of level sets and critical sets of two classes of quantum control landscapes.

The first of these families of landscapes consists of the unitary quantum control landscapes $J(U) = |\Tr(AA^{\dag}W^{\dag}U)|^{2}$ where $A$ is any complex $N\times N$ matrix and $W$ is a target unitary transformation (i.e., quantum gate).  These landscapes generalize the standard fidelity measure used in quantum information applications.  The critical point set of such a landscape in the case where $A$ is non-singular can be shown (aside from the $J=0$ minimum set) to comprise disjoint manifolds, each of which is a direct product of complex Grassmannians \cite{Ho2009,Dominy2011b}.  If $AA^{\dag}$ has eigenvalues with multiplicities $n_{1},\dots,n_{\kappa}$, then these critical submanifolds are of the form $\Gr{\nu_{1}}{n_{1}}\oplus \cdots\oplus \Gr{\nu_{\kappa}}{n_{\kappa}}$ for some set of parameters $0\leq \nu_{j}\leq n_{j}$.  It should be noted that for any such critical submanifold $C\subset\SUN$ and any root of unity $e^{2\pi i k/N}$, the set $e^{2\pi i k/N}C$ is another critical submanifold, disjoint from $C$, with the same critical value as $C$.

The second landscape family can either be described on the special unitary group by $J(U) = \Tr(U\rho_{0}U^{\dag}\mathcal{O})$ or on the unitary adjoint orbit through $\rho_{0}$ [i.e. the flag manifold $\FM$] by $J'(\rho) = \Tr(\rho\mathcal{O})$.  In either case, $\mathcal{O}$ is a Hermitian matrix describing a particular quantum mechanical observable, and the landscape represents the value that would be measured on a quantum system with state $U\rho_{0}U^{\dag}$ (or $\rho$, respectively).  It may be shown that the natural action of $\SUN$ on $\FM$ is a Hamiltonian action and that $J'$ is the associated ``Hamiltonian function'' that generates this action and defines the moment map.  As such, $J'$ is a Morse-Bott function such that each level set is connected, and therefore $J'$ has exactly one local minimum submanifold and one local maximum submanifold, i.e. $J'$ has no ``traps'' \cite{Atiyah1982, Guillemin1982}\cite[\S 3.3--3.4]{Nicolaescu2007}\cite[Ex. 3.24]{Banyaga2004}.  Furthermore, the critical submanifolds of $J'$ are diffeomorphic to direct products of generalized complex flag manifolds: $C' = \FM[k_{11},\dots, k_{1s}]\oplus\cdots\oplus\FM[k_{\kappa 1},\dots,k_{\kappa s}]$ where $\sum_{j=1}^{s}k_{ij} = n_{i}$ for each $i$.  As is the case with the adjoint orbit itself, these critical submanifolds are homogeneous K\"ahler manifolds \cite{Borel1954}.  When formulated on the unitary group $\UN$, $J$ was shown to have critical submanifolds diffeomorphic to $\big(\Unm\big)/G_{\pi}$, where $G_{\pi}$ is a subgroup of the form $\{V\oplus\pi V\pi^{\dag}\;:\; V\in\Un\text{ and } \pi V\pi^{\dag}\in\Um\}$ for some permutation matrix $\pi$, $\Un = \mathrm{U}(n_{1})\oplus\cdots\oplus\mathrm{U}(n_{\kappa})$ and similarly for $\Um$, and where $\mathbf{n} = [n_{1},\dots,n_{\kappa}]^{\rmT}$ and $\mathbf{m} = [m_{1},\dots,m_{r}]^{\rmT}$ are the multiplicities of the eigenvalues of $\rho$ and $\mathcal{O}$, respectively \cite{Wu2008}.  This is the fiber bundle defined by the restriction of $f:\UN\to\FM$ to $f^{-1}(C')$ with base $C'$ a critical submanifold of $J'$ and fiber $\mathrm{U}(n_{1})\oplus\dots\oplus\mathrm{U}(n_{\kappa})$.  Likewise the critical set of $J$ over $\SUN$ comprises submanifolds diffeomorphic to the fiber bundles given by the restriction of $g:\SUN\to\FM$ to $g^{-1}(C')$ with base $C'$ and fiber $\mathrm{S}\big(\mathrm{U}(n_{1})\oplus\dots\oplus\mathrm{U}(n_{\kappa})\big)$.

\bigskip

\subsection{Homotopy}
For any subset $\KSS$ of $\MFD$, $\EP_{x}$ restricted to $\EP_{x}^{-1}(\KSS)$ is a Hurewicz fibration with the same fiber as $\EP_{x}$ (up to homotopy equivalence)\cite{Hatcher2002}.  So assuming that $\KSS$ is connected, one way to examine the topology of $\EP_{x}^{-1}(\KSS)$ is through the long exact sequence
\begin{widetext}
\begin{equation}\cdots \!\to\! \pi_{2}(\EP_{x}^{-1}(\KSS)) \!\to\! \pi_{2}(\KSS) \!\to\! \pi_{1}(F) \!\to\! \pi_{1}(\EP_{x}^{-1}(\KSS)) \!\to\! \pi_{1}(\KSS) \!\to\! \pi_{0}(F) \!\to\! \pi_{0}(\EP_{x}^{-1}(\KSS)) \!\to\!  0,\label{eqn:subFibrationLongExactSequence}\end{equation}
\end{widetext}
which, because the fiber $F = \EP_{x}^{-1}(x)$ (the set of controls steering from $x$ back to $x$) is homotopy equivalent to $\Omega \MFD$ by Theorem \ref{thm:homotopyEquiv}, may be rewritten
\begin{widetext}
\begin{equation}\cdots \!\to\! \pi_{2}(\EP_{x}^{-1}(\KSS)) \!\to\! \pi_{2}(\KSS) \!\to\! \pi_{2}(\MFD) \!\to\! \pi_{1}(\EP_{x}^{-1}(\KSS)) \!\to\! \pi_{1}(\KSS) \!\to\! \pi_{1}(\MFD) \!\to\! \pi_{0}(\EP_{x}^{-1}(\KSS)) \!\to\!  0.\label{eqn:subFibrationLongExactSequence2}\end{equation}
\end{widetext}
This long exact sequence may be used to obtain higher order homotopy information about $\EP_{x}^{-1}(\KSS)$.  For example, when  the fiber $F$ is connected and simply connected (as is the case when $\MFD = \SUN$), the fundamental group $\pi_{1}(\EP_{x}^{-1}(\KSS))$ is isomorphic to $\pi_{1}(\KSS)$.  On the other hand, when $\MFD = \CPN$, the many trivial homotopy groups of the fiber indicated in \eqref{eqn:CPNfiberHomotopy} imply that $\pi_{i}(\EP_{x}^{-1}(\KSS)) \simeq \pi_{i}(\KSS)$ for all $i = 3,\dots,2N-3$.  Much information about the structure of these dynamical sets $\EP_{x}^{-1}(\KSS)$ can be obtained through a careful analysis of this long exact sequence and utilizing other tools of homotopy theory, including perhaps spectral sequences.  The implications for connectedness of these sets will be considered presently, however an exhaustive analysis of the higher homotopy structure is outside the scope of this paper.

\bigskip

\subsection{Connectedness}
An immediate consequence of the homotopy long exact sequence \eqref{eqn:subFibrationLongExactSequence} is that if $\KSS$ is connected and $\pi_{1}(\KSS)\to\pi_{0}(F)$ is surjective, then $\EP_{x}^{-1}(\KSS)$ is connected.  This map is defined in a similar fashion to $\eta$ from \eqref{eqn:fibrationSequence}: given an element of $\pi_{1}(\KSS)$, choose a representative loop in $\KSS$, lift it to $\EP_{x}^{-1}(\KSS)$ by the lifting function, and select the connected component of the endpoint of the resulting curve to obtain the corresponding element of $\pi_{0}(F)$.  Since $\pi_{0}(F) \simeq \pi_{0}(\Omega \MFD) \simeq \pi_{1}(\MFD)$, surjectivity of the map $\pi_{1}(\KSS)\to\pi_{0}(F)$ is equivalent to the statement that every element of the fundamental group of $\MFD$ has a representative loop in $\KSS$, or that $i_{*}:\pi_{1}(\KSS)\to\pi_{1}(\MFD)$ is surjective, where $i:\KSS\to\MFD$ is the inclusion map.  More generally, if $\pi_{1}(\KSS_{c})\to\pi_{0}(F)$ is surjective for each connected component of $\KSS_{c}\subset \KSS$, then $\EP_{x}^{-1}(\KSS)$ has exactly the same number of connected components as $\KSS$.  In particular, if $F$ is connected then this is satisfied trivially.  Since the fibers of the fibrations discussed in the last section are all connected, this statement is especially relevant to quantum control.

Since each level set of $J'(\rho) = \Tr(\rho \mathcal{O})$ is connected in $\FM$, and $\FM$ is simply connected, the above arguments imply that the corresponding ``dynamical'' level sets in $\CtlSp$ are also connected.  Of course, this conclusion is independent of the choice of state space.  In other words, had the control problem been defined on $\MFD = \SUN$ with the kinematic landscape $J(U) = \Tr(U\rho_{0}U^{\dag}\mathcal{O})$, the dynamical level sets are identical to those of the $\FM$ construction, and therefore are also connected.

With the unitary landscapes $J(U) = |\Tr(AA^{\dag}W^{\dag}U)|^{2}$ on $\SUN$ where $A\in\mathbb{C}^{N\times N}$ is an arbitrary parameter matrix \cite{Ho2009, Dominy2011b}, the connectedness story is different.  When $A$ is nonsingular, the maximal level set of $J$ comprises exactly $N$ discrete points, one for each global phase rotation of $W$ by an $N$'th root of unity, i.e. $\KSS = \big\{e^{\frac{2\pi i k}{N}}W\;:\;k = 0,\dots,N-1\big\}\subset \SUN$.  As a result, every level set from that maximal value down to the next critical value consists of exactly $N$ connected components (homeomorphic to spheres).  And therefore, each of the corresponding ``dynamical'' level sets comprises $N$ connected components.  This result is perhaps better understood in terms of control on the projective unitary group $\PUN$.  Since the global phase of the unitary evolution operator (or indeed the quantum state) is not physically measurable, the space of evolution operators should really be considered to be $\UN/\mathrm{U}(1)\mathbb{I}\cong\PUN$.  This is a Lie group with fundamental group $\pi_{1}\big(\PUN\big)\cong\mathbb{Z}/N\mathbb{Z}$, i.e. the additive group of integers mod $N$.  Since the landscape $J$ is invariant to global phase, it is well-defined on $\PUN$ and the maximal level set of controls is the same as when the control problem is defined on $\SUN$.  But the interpretation now is that, because $\pi_{1}(\PUN)\cong\mathbb{Z}/N\mathbb{Z}$, for any given target $W$ there are exactly $N$ homotopy classes of paths joining the identity $\mathbb{I}$ to the target $W$, and each class is uniquely identified with one connected component of the maximal level set of $J$ in $\CtlSp$.  This then leads to the possibility of a physical interpretation of control set topology because it may be reasonable to identify these $N$ homotopy classes of paths with $N$ classes of physical mechanisms of control for arriving at the target operator $W$.  Then each connected component of the maximal level set in $\CtlSp$ would uniquely correspond to one of these $N$ mechanism classes.

\bigskip

\subsection{(Co)Homology}
Since $\EP_{x}$ restricted to $\EP_{x}^{-1}(\KSS)$ is a Hurewicz fibration with the same fiber as $\EP_{x}$, the most promising method for evaluating the (co)homology of the dynamical set $\EP_{x}^{-1}(\KSS)$ is the Serre spectral sequence.  The Serre spectral sequence applies to Hurewicz fibrations and attempts to compute the (co)homology of the total space [in this case $\EP_{x}^{-1}(\KSS)$] given the (co)homologies of the fiber $F\simeq \Omega\MFD$ and the base $\KSS$.  The sequence appears in various forms, the simplest of which is dependent on the satisfaction of one additional assumption, namely that the action of the fundamental group of the base $\KSS$ on the (co)homology of the fiber is trivial \cite{HatcherSSAT, McCleary2001}.  Since the state spaces we have considered -- $\SUN$, $\CPN$, and $\FM$ -- are all simply connected, every loop in these spaces acts trivially on the (co)homology of the fiber.  For any subset $\KSS$ of one of these spaces, a loop in $\KSS$ is a loop in the state space and therefore acts trivially on the (co)homology of the fiber.  So, even though $\KSS$ may not be simply connected, its fundamental group acts trivially on the (co)homology of the fiber simply by virtue of $\KSS$ being a subset of a simply connected space, and the fibration being defined by restriction to $\EP_{x}^{-1}(\KSS)$.  As a result, the Serre spectral sequence beginning with $E_{p,q}^{2} \cong H_{p}(\KSS; H_{q}(F; G))$ [or $E_{p,q}^{2}\cong H^{p}(\KSS; H^{q}(F;G))$] converges to $H_{*}\big(\EP_{x}^{-1}(\KSS); G\big)$ [$H^{*}\big(\EP_{x}^{-1}(\KSS); G\big)$, respectively] for every abelian group $G$ of coefficients \cite[Thm. 5.4]{McCleary2001}\cite{HatcherSSAT}.

While this gives the starting point of the spectral sequence (assuming the (co)homology of $\KSS$ is known), the effort involved in working out the limit of the sequence can be substantial.  For this reason, a full accounting of the (co)homologies of, say, the dynamical critical point sets of $\Tr(\rho\mathcal{O})$, is beyond the scope of this paper.  However, in certain special cases, the Serre spectral sequence converges immediately at the $E^{2}$ step.  In particular, this occurs when the (co)homologies of both the fiber $F$ and the base $\KSS$ are torsion free and zero for odd dimensions.  Under these conditions, the integral homology of the total space $\EP_{x}^{-1}(\KSS)$ is simply given by $H_{*}\big(\EP_{x}^{-1}(\KSS); \mathbb{Z}\big) \cong H_{*}(\KSS; \mathbb{Z})\otimes H_{*}(F; \mathbb{Z})$, and similarly for the integral cohomology \cite{HatcherSSAT}.  

This enables a vast simplification of the problem of computing the integral cohomology of the dynamical critical sets of $|\Tr(AA^{\dag}W^{\dag}U)|$.  The integral cohomologies of the complex Grassmannian manifolds $\Gr{\nu}{n}$ and of the loop space of the special unitary group $\Omega \SUN$ are free of torsion and zero for odd dimension \cite{Bott1956} as desired for the simplification indicated above.  Because a critical submanifold $C$ of $|\Tr(AA^{\dag}W^{\dag}U)|$ in $\SUN$ is a direct product of Grassmannians, the K\"unneth formula may be invoked to conclude $H^{*}(C;\mathbb{Z})\cong H^{*}\big(\Gr{\nu_{1}}{n_{1}}\big)\otimes\cdots\otimes H^{*}\big(\Gr{\nu_{\kappa}}{n_{\kappa}}\big)$.  It follows that the integral cohomology of a dynamical critical set corresponding to the preimage of a critical submanifold $C$ of $|\Tr(AA^{\dag}W^{\dag}U)|$ in $\SUN$ is torsion-free, zero for odd dimensions, and is given by 
\begin{widetext}
\begin{equation}H^{*}\big(\EP_{\mathbb{I}}^{-1}(C);\mathbb{Z}\big) \cong H^{*}\big(\Gr{\nu_{1}}{n_{1}}\big)\otimes\cdots\otimes H^{*}\big(\Gr{\nu_{\kappa}}{n_{\kappa}}\big)\otimes \mathbb{Z}[x_{2},\dots,x_{2N-2}],\end{equation}
\end{widetext}
where the $x_{2j}$ are free generators of degree $2j$, and the cohomologies of the individual Grassmannian manifolds are given by the Poincare polynomial \cite{Bott1982} \begin{widetext}
\begin{equation}P_{t}[\Gr{\nu}{n}] = \frac{(1-t^{2})\cdots(1-t^{2n})}{(1-t^{2})\cdots(1-t^{2\nu})(1-t^{2})\cdots(1-t^{2(n-\nu)})}.\end{equation}
\end{widetext}

\bigskip

\section{Conclusions}
\label{sec:conclusion}
A proof of a generalization of Sarychev's theorem was presented, demonstrating that the space of controls driving the dynamics of an affine control system from a given starting point to a given target point is homotopy equivalent to the loop space of the state manifold of the system, provided the non-drift part of the control system satisfies the Lie algebra rank condition.  Moreover, the proof shows that the endpoint map $\EP_{x}$ is a Hurewicz fibration.  These results enable the analysis of the topological structure of subsets of control space of the form $\EP_{x}^{-1}(\KSS)\subset\CtlSp$.  Both the homotopy and (co)homology structures of critical sets and level sets may be examined using tools of algebraic topology such as exact sequences and spectral sequences, as demonstrated in the simplest cases.  It should be stressed that these results depend heavily on the control space being considered.  In particular, the proof does not apply in the case of scalar control and/or a fixed final time $T$, which have often been the conditions in prior quantum control landscape analyses.  However, in the case of a two-level quantum system, Montgomery \cite{Montgomery1996} has essentially shown how Smale's work on regular homotopy \cite{Smale1958} can be applied to a scalar control problem on $\mathrm{SO}(3)\simeq \mathrm{PU}(2)$, leading to a similar conclusion to Theorem \ref{thm:homotopyEquiv} (see Appendix \ref{sec:dynHomQubits}).  Similar analysis for scalar control problems on higher-dimensional quantum systems may be more difficult, possibly due in part to the presence of singular controls \cite{Wu2012}.

The most straight-forward topological information that can be obtained is the connectedness of these sets.  It was shown that, since the state spaces of greatest interest in closed system quantum control are generally simply connected, a ``dynamical set'' of the form $\EP_{x}^{-1}(\KSS)$ has exactly one connected component for each connected component of the ``kinematic set'' $\KSS$.  This implies, for example, that the dynamical level sets of the ensemble landscape $\Tr(\rho\mathcal{O})$ corresponding to a quantum mechanical observable $\mathcal{O}$ are all connected and that the highest dynamical level sets (from the maximum value down to the next highest critical value) each comprise exactly $N$ connected components for an $N$-level system.  This structure of the dynamical level sets and critical sets can be important for understanding the behavior of optimal control algorithms employed both in the laboratory and in numerical simulations, especially continuous level set and critical set exploration algorithms such as D-MORPH \cite{Rothman2005, Rothman2005a, Rothman2006, Beltrani2007, Dominy2008}.  In some cases, for example the case of the set of all controls producing a target unitary evolution operator up to global phase, the connected components of the control set may also have a physical interpretation in terms of control mechanism or other characteristics.  In addition, this connectedness story as well as the wealth of additional homotopy and (co)homology content of the theorem offer a tantalizing glimpse of a deeper structure of dynamical quantum control landscapes and more generally of quantum optimal control theory, the full implications of which remain to be understood.

\bigskip

\begin{acknowledgments}
This work was supported, in part, by U.S. Department of Energy (DOE) Contract No. DE-AC02-76-CHO-3073 through the Program in Plasma Science and Technology at Princeton.  We also acknowledge support from the National Science Foundation (NSF) grant No. CHE-0718610 and from the U.S. Army Research Office (ARO) grant No. W911NF-09-1-0482.
\end{acknowledgments}

\bigskip

\appendix
\section{Continuity of the Endpoint Map\label{sec:contEndpoint}}
In this appendix, we prove the claim that the endpoint map $\EP_{x}:\CtlSp\to\MFD$ is continuous.  From \cite[Theorem 10, Chapter 4]{Jurdjevic1997} comes the result that for an affine control system on a Riemannian manifold, convergence of controls in $L^{1}_{m}[0,T]$ (for a fixed final time $T$) implies uniform converges of the trajectories and therefore convergence of the endpoints.  Suppose that, in $\CtlSp$, $C_{k} = (T_{k},\CTL_{1}^{k},\dots,\CTL_{m}^{k})\to (T,\CTL_{1},\dots,\CTL_{M}) = C$.  Then $T_{k}\to T$, and in $L^{1}[0,\hat{T}]$ for $\hat{T} = \sup\{T_{k}\}<\infty$, we have that $\CTL_{j}^{k}\to\CTL_{j}$ for each $j$.  Then  the trajectories converge uniformly $z_{k}(\cdot)\to z(\cdot)$ on $[0,\hat{T}]$.  Each $z_{k}$ is the controlled trajectory over $[0,T_{k}]$ followed by pure drift (corresponding to zero control) on $[T_{k},\hat{T}]$.  Then $\sup_{[0,\hat{T}]}d(z_{k}(t),z(t))\to 0$ as $k\to \infty$, where $d(\cdot,\cdot)$ denotes the topological metric on $\MFD$ induced the Riemannian metric.  Now,
\begin{widetext}
\begin{subequations}
\begin{align}
	\sup_{\tau\in[0,1]}d(z_{k}(\tau T_{k}), z(\tau T)) &\leq \sup_{\tau\in[0,1]}d(z_{k}(\tau T_{k}), z(\tau T_{k})) + \sup_{\tau\in[0,1]}d(z(\tau T_{k}), z(\tau T))\\
	& \leq \sup_{t\in[0,\hat{T}]}d(z_{k}(t), z(t)) + \sup_{\tau\in[0,1]}d(z(\tau T_{k}), z(\tau T))\label{eqn:unifReparamBound}.
\end{align}
\end{subequations}
\end{widetext}
The first term of \eqref{eqn:unifReparamBound} vanishes as $k\to\infty$ due to uniform convergence of $z_{k}$ to $z$ on $[0,\hat{T}]$.  Since $z$ is continuous on $[0,\hat{T}]$, by the Heine(-Cantor) theorem it is uniformly continuous, so that the second term also vanishes as $k\to\infty$ and $T_{k}\to T$.  As a result, the trajectories $z_{k}$ on $[0,T_{k}]$, linearly reparametrized to $[0,1]$, converge uniformly to $z$, similarly reparametrized.  This means first that the map from $\CtlSp$ to this linearly reparametrized trajectory space with uniform topology is a continuous map, so the metric topology we have defined on $\CtlSp$ is strictly stronger than this uniform topology on trajectory space.  And second, that the endpoint map $\EP_{x}$ is continuous.

\bigskip

\section{Continuity of Control Concatenation\label{sec:contControlConcat}}
We defined concatenation of controls in \eqref{eqn:controlConcat} with notation $C\star D$ for $C,D\in\CtlSp$.  In this section we prove that this operation $\star:\CtlSp\times\CtlSp\to\CtlSp$ is a continuous map.  Since $\CtlSp$ is a metric space, $\CtlSp\times\CtlSp$ is also metric, so it suffices to consider the convergence of $C_{k}\star D_{k}$ when $(C_{k},D_{k})\to(C,D)$ and therefore $C_{k}\to C$ and $D_{k}\to D$.  We will write the components of these controls as $C = (T^{C},\CTL_{1},\dots, \CTL_{M})$, $C_{k} = (T_{k}^{C}, \CTL_{1}^{k},\dots, \CTL_{M}^{k})$, $D = (T^{D},\CTLB_{1},\dots,\CTLB_{M})$ and $D_{k} = (T_{k}^{D},\CTLB_{1}^{k},\dots,\CTLB_{M}^{k})$.  By the triangle inequality, $d(C_{k}\star D_{k},C\star D) \leq d(C_{k}\star D_{k}, C_{k}\star D) + d(C_{k}\star D, C\star D)$.  It may easily be observed from the definition of $\star$ and of the metric $d$ on $\CtlSp$, that $d(C_{k}\star D_{k}, C_{k}\star D) = d(D_{k},D)$.  Now 
\begin{widetext}
\begin{subequations}
\begin{align}
	d(C_{k}\star D, C\star D) & = \begin{cases}|T_{k}^{C} - T^{C}| + \sum_{j=1}^{M}\int_{0}^{T^{C}}|\CTL_{j}^{k}(t) - \CTL_{j}(t)|\,\rmd t & \\
	 \qquad + \int_{T^{C}}^{T_{k}^{C}}|\CTL_{j}^{k}(t) - \CTLB_{j}(t-T^{C})|\,\rmd t & \\
	 \qquad + \int_{T_{k}^{C}}^{\infty}|\CTLB_{j}(t-T_{k}^{C}) - \CTLB_{j}(t-T^{C})|\,\rmd t & T^{C} \leq T_{k}^{C}\\
	|T_{k}^{C} - T^{C}| + \sum_{j=1}^{M}\int_{0}^{T_{k}^{C}}|\CTL_{j}^{k}(t) - \CTL_{j}(t)|\,\rmd t &\\
	 \qquad + \int_{T_{k}^{C}}^{T^{C}}|\CTL_{j}(t) - \CTLB_{j}(t-T_{k}^{C})|\,\rmd t & \\
	 \qquad + \int_{T^{C}}^{\infty}|\CTLB_{j}(t-T_{k}^{C}) - \CTLB_{j}(t-T^{C})|\,\rmd t & T_{k}^{C} \leq T^{C}
	 \end{cases}\\
	 & \leq d(C_{k},C) + \sum_{j=1}^{M}\int_{0}^{|T_{k}^{C} - T^{C}|}|\CTLB_{j}(t)|\,\rmd t\nonumber\\
	 & \qquad + \int_{0}^{T^{D}}|\CTLB_{j}(t + |T_{k}^{C}-T^{C}|) - \CTLB_{j}(t)|\,\rmd t.
\label{eqn:concatBound}\end{align}
\end{subequations}
\end{widetext}

Since the continuous functions are dense in $L^{1}$, let $\{\mathcal{G}_{j}^{l}\}$ be continuous functions supported on $[0,T^{D} + 1/l]$ such that $\mathcal{G}_{j}^{l} \to\CTLB_{j}$ on $[0,\infty)$.  Then for each $l$, 
\begin{widetext}
\begin{equation}\lim_{k\to\infty}\int_{0}^{|T_{k}^{C} - T^{C}|}|\CTLB_{j}(t)|\,\rmd t \leq \|\CTLB_{j} - \mathcal{G}_{j}^{l}\| + \lim_{k\to\infty}\int_{0}^{|T_{k}^{C} - T_{C}|}|\mathcal{G}_{j}^{l}(t)|\,\rmd t = \|\CTLB_{j} - \mathcal{G}_{j}^{l}\|.\label{eqn:limMiddleConcatTerm}\end{equation}
\end{widetext}
For any $\epsilon>0$ there exists $L$ such that for all $l>L$, $\|\CTLB_{j} - \mathcal{G}_{j}^{l}\|<\epsilon$, so we see that the left hand side of \eqref{eqn:limMiddleConcatTerm} is zero.  Likewise, the third term of \eqref{eqn:concatBound} can be bounded for each $l$ as 
\begin{widetext}
\begin{equation}\lim_{k\to\infty}\int_{0}^{T^{D}}|\CTLB_{j}(t + |T_{k}^{C}-T^{C}|) - \CTLB_{j}(t)|\,\rmd t \leq 2\|\CTLB_{j} - \mathcal{G}_{j}^{l}\| + \lim_{k\to\infty}\int_{0}^{T^{D}+ \frac{1}{l}}|\mathcal{G}_{j}^{l}(t + |T_{k}^{C} - T^{C}|) - \mathcal{G}_{j}^{l}(t)|\,\rmd t.\label{eqn:limEndConcatTerm}\end{equation}
\end{widetext}
Since $\mathcal{G}_{j}^{l}$ is a continuous real-valued function on a compact interval, by the Heine(-Cantor) theorem, it is uniformly continuous, so for each $\epsilon >0$, there exists $K$ such that for all $k>K$, $|\mathcal{G}_{j}^{l}(t+|T_{k}^{C}-T^{C}|) - \mathcal{G}_{j}^{l}(t)| \leq \epsilon$ for all $t\in [0,T^{D}+1/l]$.  So the integral on the right hand side of \eqref{eqn:limEndConcatTerm} is bounded by $\epsilon(T^{D} + 1/l)$ for any arbitrarily small $\epsilon>0$ and so it vanishes in the limit as $k\to \infty$, meaning that the left hand side is bounded by $2\|\CTLB_{j} - \mathcal{G}_{j}^{l}\|$ for any $l$.  Since this can be made as small as desired, the left hand side of \eqref{eqn:limEndConcatTerm} must vanish.  So we finally conclude that 
\begin{equation}d(C_{k}\star D_{k},C\star D) \leq d(C_{k}\star D_{k}, C_{k}\star D) + d(C_{k}\star D, C\star D) \to 0\end{equation} so that concatenation is a continuous operation in the $\CtlSp$ metric.

\bigskip

\section{Uniform Continuity of the Cross-Section Map\label{sec:contCrossSec}}
Recall from Lemma \ref{lem:CrossSectionMap} that the cross-section map $\sigma:\mathcal{W}\times\mathcal{W}\subset\MFD\to\CtlSp$ was defined to be $\sigma(x,y) = F(x,\varphi(x,y))$, where $\mathcal{W}$ is an open neighborhood of $x_{0}\in\MFD$, $F:\mathbb{R}^{n}\to\mathbb{K}$, and $\varphi:\mathcal{W}\times \mathcal{W}\to\mathbb{R}^{n}$.  Since $\varphi$ was defined and shown to be $C^{1}$ by the implicit function theorem, the proof that $\sigma$ is continuous reduces to showing that $F$ is continuous.  Then \emph{uniform} continuity follows from observing that, by restricting to a smaller neighborhood of $x_{0}$ if necessary, $\sigma$ is continuous on $\overline{\mathcal{W}}\times\overline{\mathcal{W}}$, which is compact.  So invoking Heine(-Cantor), we can conclude that $\sigma$ is uniformly continuous.

Turning now to the proof that $F$ is continuous, let us first consider maps \begin{equation}\hat{R}^{\nu}(\VF_{0}, Y_{1},\dots, Y_{\nu},\xi_{1},\dots,\xi_{\nu}) = (T,\CTL_{1},\dots,\CTL_{M})\in\CtlSp\end{equation} that produce the piecewise constant controls whose integral curves are exactly the sequence of exponentials in the definition of $R^{\nu}$, hence $\EP\circ\hat{R}^{\nu} = R^{\nu}$.  Then there exist $S_{k}:\mathbb{R}^{n}\to\CtlSp$ for $k = 1,\dots, 2^{\nu}+2^{\nu-1}-2$, such that for fixed $Y_{1},\dots,Y_{\nu}$, $\hat{R}^{\nu} = S_{1}\star S_{2}\star\dots\star S_{2^{\nu}+2^{\nu-1}-2}$, and where each $S_{k}(\vec{\xi}) = (T_{k},\CTLB_{1}^{k},\dots,\CTLB_{M}^{k})$ depends on only one of the $\xi_{j}$'s, is defined on an interval of length $T_{k} = \xi_{j}^{2\alpha_{\nu}}$, and each $\CTLB_{i}^{k}$ takes on a constant value proportional to $\xi_{j}^{1 - 2\alpha_{\nu}}$.  Now, it is easy to see that each $S_{k}$ is continuous in $\vec{\xi}$.  So, by the continuity of $\star$ proved in Appendix \ref{sec:contControlConcat}, $\hat{R}^{\nu}$ depends continuously on $\vec{\xi}$.  Then, since $F(r)$ is formed by concatenating the output of these $\hat{R}^{\nu}$ functions with the input of each $\hat{R}^{\nu}$ depending continuously on some $y_{j}$, we have that $F$ is continuous.  And by the above arguments, we conclude that $\sigma$ is uniformly continuous.

\bigskip

\section{Continuity of the Lifting Function\label{sec:contLifting}}
In this appendix, we prove that the (local) lifting function $\lambda_{\mathcal{W}}:X_{\mathcal{W}}\to P\CtlSp$ is a continuous map in the appropriate topologies.  Recall that $\mathcal{W}\subset\MFD$ is an open set and $X_{\mathcal{W}}\subset \EP_{b}^{-1}(\mathcal{W})\times P\mathcal{W}$ consists of pairs $(C,\gamma)$ such that $\EP_{b}(C) = \gamma(0)$.  $P\mathcal{W}$ is given the relative topology as a subset of the path space $P\MFD$, which itself is given the compact-open topology.  Since $\MFD$ is a metric space under the bi-invariant topological metric induced by the Riemannian metric from the real Hilbert-Schmidt inner product, both $P\MFD$ and $P\mathcal{W}$ are also metric spaces.  Likewise $P\CtlSp$, $\EP^{-1}(\mathcal{W})$, and $X_{\mathcal{W}}$ are all metric spaces.  So to prove continuity of $\lambda_{\mathcal{W}}$, it suffices to prove convergence of $\lambda_{\mathcal{W}}(C_{k},\gamma_{k})$ to $\lambda_{\mathcal{W}}(C,\gamma)$ when $(C_{k},\gamma_{k})\to(C,\gamma)$.

Recall that $\lambda_{\mathcal{W}}$ was defined by $\lambda_{\mathcal{W}}(C,\gamma) = C\star\sigma(\EP_{b}(C),\gamma(s))$, where $\sigma$ is the cross-section map defined in Lemma \ref{lem:CrossSectionMap} whose uniform continuity was established in Appendix \ref{sec:contCrossSec}.  By the triangle inequality and the definitions of the metric on $\CtlSp$ and of the concatenation operator $\star$, 
\begin{widetext}
\begin{subequations}
\begin{align}
	d\big(\lambda_{\mathcal{W}}(C_{k},\gamma_{k}),\lambda_{\mathcal{W}}(C,\gamma)\big) & = \sup_{s\in[0,1]}d\big(C_{k}\star\sigma(\EP_{b}(C_{k}), \gamma_{k}(s)), C\star\sigma(\EP_{b}(C), \gamma(s))\big)\\
	 & \leq  \sup_{s\in[0,1]}d\big(C_{k}\star\sigma(\EP_{b}(C_{k}), \gamma_{k}(s)), C_{k}\star\sigma(\EP_{b}(C), \gamma(s))\big)\nonumber\\
	 & \qquad  + \sup_{s\in[0,1]}d\big(C_{k}\star\sigma(\EP_{b}(C), \gamma(s)), C\star\sigma(\EP_{b}(C), \gamma(s))\big)\\
	 & = \sup_{s\in[0,1]}d\big(\sigma(\EP_{b}(C_{k}), \gamma_{k}(s)), \sigma(\EP_{b}(C), \gamma(s))\big)\nonumber\\
	 & \qquad  + \sup_{s\in[0,1]}d\big(C_{k}\star\sigma(\EP_{b}(C), \gamma(s)), C\star\sigma(\EP_{b}(C), \gamma(s))\big).\label{eqn:liftingTriangleIneq}
\end{align}
\end{subequations}
\end{widetext}
Consider the first terms on the right hand side of \eqref{eqn:liftingTriangleIneq}.  Since by assumption $\gamma_{k}\to\gamma$ uniformly and $C_{k}\to C$, it follows by continuity of $\EP_{b}$ that $\EP_{b}(C_{k})\to\EP_{b}(C)$, and therefore $(\EP_{b}(C_{k}), \gamma_{k}(s))\to(\EP_{b}(C),\gamma(s))$ uniformly over $s\in[0,1]$.  Then the uniform continuity of $\sigma$ implies that $\sigma(\EP_{b}(C_{k}), \gamma_{k}(s))\to \sigma(\EP_{b}(C), \gamma(s))$ uniformly, so this first term converges to zero as $k\to\infty$.

Turning now to the second term on the right hand side of \eqref{eqn:liftingTriangleIneq}, observe that $\overline{\{C_{k}\}}$ is a compact subset of $\CtlSp$ since any open cover must admit an open set containing $C$ and this set will also contain all but finitely many of the $C_{k}$'s, so that only finitely many additional sets from the cover are needed to construct a finite subcover.  Also $\{\sigma(\EP_{b}(C), \gamma(s))\;:\; s\in[0,1]\}$ is a compact subset of $\CtlSp$ since it is the continuous image of a compact set.  Therefore $A:= \overline{\{C_{k}\}}\times\{\sigma(\EP_{b}(C), \gamma(s))\;:\;s\in[0,1]\}$ is compact in $\CtlSp\times\CtlSp$.  Since $\star$ is continuous over $\CtlSp\times\CtlSp$, it is uniformly continuous over $A$ by the Heine(-Cantor) theorem.  So for any $\epsilon>0$ there exists a $\delta>0$ such that for any $(\alpha_{1},\beta_{1}), (\alpha_{2},\beta_{2})\in A$ with $d(\alpha_{1},\alpha_{2})<\delta$ and $d(\beta_{1},\beta_{2})<\delta$, $d(\alpha_{1}\star\beta_{1}, \alpha_{2}\star\beta_{2})<\epsilon)$.  Then for any $\epsilon>0$, there exists $K$ such that for all $k>K$, $d(C_{k},C)<\delta$, and therefore $d\big(C_{k}\star \sigma(\gamma(s)\EP(C)^{\dag}), C\star\sigma(\gamma(s)\EP(C)^{\dag})\big)<\epsilon$ for all $s\in[0,1]$.  Therefore the second term on the right hand side of \eqref{eqn:liftingTriangleIneq} converges to zero, so the left hand side converges to zero, and therefore the lifting function $\lambda_{\mathcal{W}}$ is continuous with respect to the given topologies.

\bigskip

\section{Dynamic Homotopy of Scalar Control on 2-Level Systems\label{sec:dynHomQubits}}

Consider a right-invariant affine control system of the form \begin{equation}i\frac{\rmd U}{\rmd t} = \big(H_{0} + \CTL(t)H_{1}\big)U(t) \qquad \qquad U(0) = \mathbb{I}\label{eqn:scalarControlSystem}\end{equation}
on the $2\times 2$ special unitary group $\SU(2)$, hence $H_{0}$ and $H_{1}$ are Hermitian and trace zero.  We will assume that the system satisfies the Lie algebra rank condition (LARC), so that $[iH_{0},iH_{1}]$ lies outside the linear span of $\{iH_{0}, iH_{1}\}$.  In this appendix, we describe the relationship between the topological study of the trajectory space of this scalar control problem and Smale's work on regular homotopy.

For a given control $\CTL$, let 
\begin{subequations}
\begin{align}
	A_{0} & := \frac{iH_{1}}{\|H_{1}\|} & A(t) := U^{\dag}(t)A_{0}U(t)\\
	B_{0} & := \frac{[iH_{0}, iH_{1}]}{\|[iH_{0}, iH_{1}]\|} & B(t) := U^{\dag}(t)B_{0}U(t)\\
	C_{0} & := \frac{iH_{0} - \langle iH_{0}, A_{0}\rangle A_{0}}{\|iH_{0} - \langle iH_{0}, A_{0}\rangle A_{0}\|}  & C(t) := U^{\dag}(t)C_{0}U(t)
\end{align}
\end{subequations}
so that $A(t)$, $B(t)$, and $C(t)$ are curves on the unit sphere within the three dimensional $\su(2)$.  Moreover, $A_{0}$, $B_{0}$, and $C_{0}$ may be seen to comprise an orthonormal basis for $\su(2)$ under the Hilbert-Schmidt inner product, so that $A(t)$, $B(t)$, and $C(t)$ are a rotating orthonormal frame.  Then it is readily found that
\begin{subequations}
\begin{align}
	\frac{\rmd A}{\rmd t}(t) & = U^{\dag}(t)[iH(t), A_{0}]U(t) = \alpha B(t)\\
	\frac{\rmd B}{\rmd t}(t) & = U^{\dag}(t)[iH(t), B_{0}]U(t) = -\alpha A(t) + \big(\beta + \gamma\CTL(t)\big) C(t)\\
	\frac{\rmd C}{\rmd t}(t) & = U^{\dag}(t)[iH(t), C_{0}]U(t) = -\big(\beta + \gamma\CTL(t)\big)B(t)
\end{align}
where \begin{equation}\alpha := \|[iH_{0}, A_{0}]\| \qquad\qquad \beta := \langle iH_{0},A_{0}\rangle \|[C_{0},A_{0}]\| \qquad \qquad \gamma := \|[C_{0}, iH_{1}]\|.\end{equation}
\end{subequations}
So $A(t)$ travels at constant speed $\alpha$ over $S^{2}\subset\su(2)$ with continuous velocity vector $\alpha B(t)$.  Therefore, $A(t)$ is a \emph{regular curve} \cite{Smale1958} on $S^{2}\subset \su(2)$.  The covariant derivative of $B(t)$ is just $\frac{D B}{dt} = \big(\beta + \gamma \CTL(t)\big)C(t)$ so that $B(t)$ can change arbitrarily quickly in either direction ($C$ or $-C$) by tuning the control $\CTL(t)$.  As a result, \emph{any} regular curve in $S^{2}\subset\su(2)$ starting at $A_{0}$ with initial normalized velocity $B_{0}$ and with constant speed $\alpha$ may be generated in this way by choosing the correct control function $\CTL$ and final time $T$.  In other words, since the curves $A(t)$, $B(t)$, and $C(t)$ uniquely determine the propagator in $\PU(2)$, there is a one-to-one correspondence between trajectories of \eqref{eqn:scalarControlSystem} on $\PU(2)$ and regular curves on $S^{2}\in\su(2)$ starting at $A_{0}$ (up to reparametrization).  This construction is essentially the same as the example described by Montgomery \cite{Montgomery1996} relating control trajectories on $\SO(3)$ ($\simeq \PU(2)$) to regular curves on $S^{2}$.

	Stephen Smale studied regular homotopies of regular curves on Riemannian manifolds \cite{Smale1958}. In particular, he looked at the set $E(p_{1}, v_{1}, p_{2}, v_{2})$ of regular curves starting at a point $p_{1}$ with initial velocity vector $v_{1}$ and ending at a point $p_{2}$ with final velocity $v_{2}$.  He derived a weak homotopy equivalence between $E(p_{1}, v_{1}, p_{2}, v_{2})$ on a manifold $\mathcal{M}$ and $\Omega \rmT_{0}(\mathcal{M})$, the loop space on the unit tangent bundle of the manifold, where the unit tangent bundle $\rmT_{0}(\mathcal{M})$ is the bundle of all unit length (i.e. normalized) tangent vectors to $\mathcal{M}$.  Now, the unit tangent bundle of the sphere, $\rmT_{0}(S^{2})$, can be identified with $\SO(3)$ by thinking of the first column of $O\in\SO(3)$ as the point in $S^{2}$, the second column as an arbitrary unit tangent vector to $S^{2}$ at that point, and the third column as simply the cross product of the first two, carrying no additional information.  Thus, since $\SO(3)\simeq \PU(2)$, Smale's theorem shows that $E(A(0), B(0), A(1), B(1))$ (where curves have been reparametrized to travel at constant speed on $[0,1]$, eliminating the final time $T$ as a variable) on $S^{2}\subset \su(2)$ is weakly homotopy equivalent to $\Omega \PU(2)$, the loopspace on $\PU(2)$.

	It may be observed that the topology given by Smale to the space of regular curves on $S^{2}$ is identical to that induced by the compact-open topology on the path space of $\PU(2)$, so that the one-to-one correspondence established above between trajectories of the control problem on $\PU(2)$ and regular curves on $S^{2}$ is a homeomorphism between these spaces.  This result implies the existence of a weak homotopy equivalence between the set of trajectories of \eqref{eqn:scalarControlSystem} joining $\hat{\mathbb{I}}\in\PU(2)$ to a target $\hat{W}\in\PU(2)$ and the loop space of $\PU(2)$.  Thus, we have the following theorem:
\begin{theorem}
\label{thm:projTrajHomEquiv}
Given any $\hat{W}\in\PU(2)$, the space of all trajectories from $\mathbb{I}$ to $\hat{W}$ of the control system \eqref{eqn:scalarControlSystem} is weakly homotopy equivalent to $\Omega\PU(2)$, the loop space of $\PU(2)$.  This loop space (and therefore the indicated set of trajectories) comprises two connected components.
\end{theorem}

	Moreover, using the fact that $\SU(2)$ is the universal covering group for $\PU(2)$, it is straightforward to lift this result to $\SU(2)$, yielding the theorem
\begin{theorem}
\label{thm:trajHomEquiv}
Given any $W\in\SU(2)$, the space of all trajectories from $\mathbb{I}$ to $W$ of the control system \eqref{eqn:scalarControlSystem} is weakly homotopy equivalent to $\Omega\SU(2)$, the loop space of $\SU(2)$.  In particular, this loop space (and therefore the indicated set of trajectories) is connected.
\end{theorem}
In addition, since Smale has shown \cite{Smale1958} that the endpoint map on the space of regular curves on $S^{2}$ is a Serre fibration \cite{Hatcher2002}, and since this space of curves in homeomorphic to the space of $\PU(2)$ trajectories of \eqref{eqn:scalarControlSystem}, the fact that the projection $\pi:\PU(2)\to\mathbb{C}P^{1}$ onto the Riemann sphere $\mathbb{C}P^{1}$ (i.e. the Bloch sphere; i.e. the complex projective line) is a fibration may be used to show that the endpoint map on the space of $\mathbb{C}P^{1}$ trajectories for some initial state $\rho_{0}\neq \mathbb{I}/2$ is also a Serre fibration.  Similar to the proof of theorem \ref{thm:homotopyEquiv}, using the contractibility of this space of trajectories, this implies
\begin{theorem}
\label{thm:riemannSphereTrajHomEquiv}
Given any $2\times 2$ density matrix $\rho_{0}\neq \mathbb{I}/2$, and any $\rho_{1}$ isoentropic to $\rho_{0}$ (i.e. possessing the same eigenvalues), the space of all trajectories from $\rho_{0}$ to $\rho_{1}$ of the control system
\begin{equation}i\frac{\rmd \rho}{\rmd t} = [H_{0}+\CTL(t)H_{1},\rho(t)] \qquad\qquad \rho(0) = \rho_{0}\end{equation} is weakly homotopy equivalent to $\Omega\mathbb{C}P^{1}$, the loop space of $\mathbb{C}P^{1}\simeq S^{2}$.  In particular, this loop space (and therefore the indicated set of trajectories) is connected.
\end{theorem}

It is speculated that these results could be pulled back to control space as in Theorem \ref{thm:homotopyEquiv}, however this would require redoing Smale's arguments from the control perspective and with a specific topology on the control space in mind.

\bigskip

\bibliography{DynamicHomotopy_arXiv}

\end{document}